\documentclass[12pt]{amsart}
\usepackage{mathrsfs}
\usepackage{amssymb} 
\usepackage{url}
\usepackage{hyperref}
\usepackage{graphicx} 

\theoremstyle{definition}
\newtheorem{theorem}{Theorem}[section]

\newcommand{\image}[3]{\begin{figure*}[ht]
\includegraphics[width=#2\textwidth]{#1}
\caption{\small{\label{#1}#3}}\end{figure*}}

\def\({\left(}
\def\){\right)}

\newcommand{\R}{\mathbb{R}}

\newcommand{\de}{\textnormal{d}}

\newcommand{\ds}{\displaystyle}

\newcommand{\eg}{\textit{e.g.} }

\newcommand{\Ric}{\textnormal{Ric}}

\newcommand{\mc}[1]{\mathcal{#1}}

\newcommand{\sref}[1]{\S\ref{#1}}

\newcommand{\dsfrac}[2]{\ds{\frac{#1}{#2}}}

\def\hyph{-\penalty0\hskip0pt\relax}

\newcommand{\nonsing}{non{\hyph}singular}

\newcommand{\flrw}{Friedmann-Lema\^itre-Robertson-Walker}
\newcommand{\FLRW}{FLRW}

\begin{document} 
 
\title[{Beyond the FLRW Big Bang singularity}]{Beyond the Friedmann-Lema\^itre-Robertson-Walker Big Bang singularity}

\author{Cristi \ Stoica}
\thanks{email: holotronix@gmail.com}
\date{Feb. 5, 2012 -- Jul. 27, 2012}

\begin{abstract}
Einstein's equation, in its standard form, breaks down at the Big Bang singularity. A new version, equivalent to Einstein's whenever the latter is defined, but applicable in wider situations, is proposed. The new equation remains smooth at the Big Bang singularity of the Friedmann-Lema\^itre-Robertson-Walker model. It is a tensor equation defined in terms of the Ricci part of the Riemann curvature. It is obtained by taking the Kulkarni-Nomizu product between Einstein's equation and the metric tensor.
\bigskip
\noindent 
\end{abstract}



\maketitle

\setcounter{tocdepth}{1}
\tableofcontents

\section*{Introduction}

The {\flrw} (\FLRW) spacetime represents, in conformity with General Relativity, the universe at very large scale, where it is assumed to be homogeneous and isotropic \cite{FRI22de,FRI99en,FRI24,LEM27,ROB35I,ROB35II,ROB35III,WAL37}. At the Big Bang, it exhibits a singularity, considered a breakdown of General Relativity \cite{HP70,Haw76,ASH91,HP96,Ash08,Ash09}. There the tensors involved in the Einstein equations diverge. It was hoped that this problem is a consequence of the high symmetry of the {\FLRW} solution, and it disappears in real-life situations, but it turned out that is the norm, rather than the exception, as Hawking's singularity theorem shows \cite{Haw66i,Haw66ii,Haw67iii}.

In this paper we will see that, although Einstein's equation in its standard form breaks down at the Big Bang singularity in the {\FLRW} model, it can be rewritten in a form which doesn't break down:
\begin{equation*}
	(G\circ g)_{abcd} + \Lambda (g\circ g)_{abcd} = \kappa (T\circ g)_{abcd}.
\end{equation*}
This form is equivalent to the standard form where the metric is not singular, but in addition extends smoothly at the singularities. Like Einstein's equation, this new equation is tensorial, but instead of being expressed in terms of the \textit{Ricci tensor} (of second order), it involves \textit{the Ricci part of the Riemann curvature tensor} (of order $4$). It is obtained by taking the Kulkarni-Nomizu product \eqref{eq_kulkarni_nomizu} between Einstein's equation and the metric tensor.

The proposed new version of Einstein's equation extends naturally and smoothly beyond the Big Bang singularity of the {\FLRW} model, and it is based solely on standard General Relativity, without any modifications.

In section \sref{s_flrw} we review briefly the fact that in the {\FLRW} model, Einstein's equation blows up. In \sref{s_einstein_eq_expanded} we introduce the expanded version of Einstein's equation, which is equivalent to the latter at the points where the metric is {\nonsing}, but it is more general. Section \sref{s_beyond_big_bang} contains the central result, a theorem showing that the new equation is smooth everywhere, including at the Big Bang singularity. Section \sref{s_properties} discusses some properties of the proposed equation and solution. We conclude with some observations and implications in \sref{s_conclusions}.

\section{{\FLRW} Big-Bang singularity}
\label{s_flrw}

Let $I\subseteq \R$ be an interval representing the time, with the natural metric $-c^2\de t^2$. Let $(\Sigma,g_\Sigma)$ be a three-dimensional Riemannian space, so that at any moment of time $t\in I$ the space is $\Sigma_t=(\Sigma,a^2(t)g_\Sigma)$, where $a: I\to \R$ is a function of time. The {\FLRW} spacetime is $I\times\Sigma$, with the metric
\begin{equation}
\label{eq_flrw_metric}
\de s^2 = -c^2\de t^2 + a^2(t)\de\Sigma^2.
\end{equation}
To model the homogeneity and isotropy conditions at large scale, one may take the Riemannian three-manifold $\Sigma$ to be one of the homogeneous spaces $S^3$, $\R^3$, and $H^3$. Then the metric on $\Sigma$ is, in spherical coordinates $(r,\theta,\phi)$,
\begin{equation}
\label{eq_flrw_sigma_metric}
\de\Sigma^2 = \dsfrac{\de r^2}{1-k r^2} + r^2\(\de\theta^2 + \sin^2\theta\de\phi^2\),
\end{equation}
where $k=1,0,-1$, for the $3$-sphere $S^3$, the Euclidean space $\R^3$, or hyperbolic space $H^3$ respectively.

If the universe is filled with a fluid with mass density $\rho(t)$ and pressure density $p(t)$, the stress-energy tensor is
\begin{equation}
\label{eq_friedmann_stress_energy}
T^{ab} = \(\rho + \dsfrac{p}{c^2}\)u^a u^b + p g^{ab},
\end{equation}
where $g(u,u)=-c^2$.

The mass density $\rho(t)$ is determined from $a(t)$ by the \textit{Friedmann equation}
\begin{equation}
\label{eq_friedmann_density}
\rho = \kappa^{-1}\(3\dsfrac{\dot{a}^2 + kc^2}{c^2 a^2} - \Lambda \),
\end{equation}
and pressure density $p(t)$ by the \textit{acceleration equation}
\begin{equation}
\label{eq_acceleration}
\dsfrac{p}{c^2} = \dsfrac{2}{\kappa c^2}\(\dsfrac{\Lambda}{3}-\dsfrac{1}{c^2} \dsfrac{\ddot{a}}{a}\) - \dsfrac \rho 3.
\end{equation}
Both these equations are consequence of the Einstein equation with the stress-energy tensor from equation \eqref{eq_friedmann_stress_energy}.

When $a\to 0$, the metric becomes degenerate, as we can see from equation \eqref{eq_flrw_sigma_metric}. In the same time, equations \eqref{eq_friedmann_density} and \eqref{eq_acceleration} imply that both $\rho$ and $p$ blow up. Consequently, the stress-energy tensor \eqref{eq_friedmann_stress_energy} blows up too. But the expanded stress-energy tensor $(T \circ g)_{abcd}$ is smooth, as it is the expanded Einstein equation which we propose here.

\section{Einstein's equation expanded}
\label{s_einstein_eq_expanded}

The Einstein equation
\begin{equation}
\label{eq_einstein}
	G_{ab} + \Lambda g_{ab} = \kappa T_{ab}
\end{equation}
involves the stress-energy tensor $T_{ab}$ of the matter, the \textit{cosmological constant} $\Lambda$, and the constant $\kappa:=\dsfrac{8\pi \mc G}{c^4}$, where $\mc G$ and $c$ are the gravitational constant and the speed of light. 
The Einstein tensor
\begin{equation}
\label{eq_einstein_tensor}
	G_{ab}:=R_{ab}-\frac 1 2 R g_{ab},
\end{equation}
is obtained from the \textit{Ricci curvature} $R_{ab} := g^{st}R_{asbt}$ and the \textit{scalar curvature} $R := g^{st}R_{st}$.

The \textit{expanded Einstein equation} is
\begin{equation}
\label{eq_einstein_expanded}
	(G\circ g)_{abcd} + \Lambda (g\circ g)_{abcd} = \kappa (T\circ g)_{abcd}
\end{equation}
where, for two symmetric bilinear forms $h$ and $k$,
\begin{equation}
\label{eq_kulkarni_nomizu}
	(h\circ k)_{abcd} := h_{ac}k_{bd} - h_{ad}k_{bc} + h_{bd}k_{ac} - h_{bc}k_{ad}
\end{equation}
denotes the \textit{Kulkarni-Nomizu product}.

So long as the metric $g$ is {\nonsing}, the Einstein equation and its expanded version are equivalent. But the expanded version \eqref{eq_einstein_expanded} remains {\nonsing}, and even smooth, in a wider rage of cases. In the example of the {\FLRW} spacetime, the metric becomes degenerate (its determinant cancels), the Einstein tensor $G_{ab}$ becomes singular, but we will see that the Kulkarni-Nomizu product $G\circ g$ tends to $0$ and cancels the blow up of the Einstein tensor.

Explicitly, the expanded Einstein equation \eqref{eq_einstein_expanded} can be rewritten as
\begin{equation}
\label{eq_einstein_expanded_explicit}
	2 E_{abcd} - 3 S_{abcd} + \Lambda (g\circ g)_{abcd} = \kappa (T\circ g)_{abcd},
\end{equation}
in terms of the \textit{scalar part}
\begin{equation}
	S_{abcd} = \dsfrac{1}{12}R(g\circ g)_{abcd}
\end{equation}
and the \textit{semi-traceless part}
\begin{equation}
	E_{abcd} = \dsfrac{1}{2}(S \circ g)_{abcd},
\end{equation}
of the Riemann curvature, where
\begin{equation}
\label{eq_ricci_traceless}
S_{ab} := R_{ab} - \dsfrac{1}{4}Rg_{ab}
\end{equation}
is the traceless part of the Ricci curvature.

These tensors are well-known from the Ricci decomposition of the Riemann curvature tensor:
\begin{equation}
	R_{abcd} = S_{abcd} + E_{abcd} + C_{abcd},
\end{equation}
where $C_{abcd}$ is the \textit{Weyl curvature tensor} (see \eg \cite{ST69,BESS87,GHLF04}).

The equation \eqref{eq_einstein_expanded_explicit} is obtained from \eqref{eq_einstein_expanded} and from
\begin{equation}
G_{ab} = S_{ab} - \dsfrac{1}{4}R g_{ab},
\end{equation}
because
\begin{equation}
\label{eq_einstein_tensor_expanded}
\begin{array}{lrl}
(G\circ g)_{abcd} &=& (S \circ g)_{abcd} - \dsfrac{1}{4}R (g\circ g)_{abcd}\\
&=& 2 E_{abcd} - 3 S_{abcd}.
\end{array}
\end{equation}

\section{Beyond the {\FLRW} Big-Bang singularity}
\label{s_beyond_big_bang}

\begin{theorem}
For the {\FLRW} metric \eqref{eq_flrw_metric},
 with $a: I\to \R$ a smooth function of time, the tensors $R_{abcd}$, $S_{abcd}$, and $E_{abcd}$ are smooth, and consequently the expanded Einstein equation is smooth too, even when $a(t)=0$.
\end{theorem}
\begin{proof}
If we denote by $\widetilde T_{ab}:= \kappa T_{ab} - \Lambda g_{ab}$,
\begin{equation}
\label{eq_flrw_curv_ricci}
\begin{array}{lll}
R_{ab} &=& \widetilde T_{ab} - \dsfrac{1}{2}g^{st}\widetilde T_{st} \\
&=& \kappa T_{ab} - \Lambda g_{ab} - \dsfrac{\kappa}{2} g^{st} T_{st} g_{ab} + 2 \Lambda g_{ab} \\
&=& \kappa\(\rho + \dsfrac{p}{c^2}\) u_a u_b + \kappa p g_{ab} - \dsfrac{\kappa}{2} \(-\rho c^2 - p + 4 p\) g_{ab} + \Lambda g_{ab} \\
&=& \kappa\(\rho + \dsfrac{p}{c^2}\) u_a u_b + \dsfrac{\kappa}{2} \(\rho c^2 - p\) g_{ab} + \Lambda g_{ab} \\
\end{array}
\end{equation}
From equations \eqref{eq_friedmann_density}, \eqref{eq_acceleration}, and \eqref{eq_flrw_curv_ricci}, we can see that the Ricci tensor has the form
\begin{equation}
\label{eq_flrw_curv_ricci_a}
	R_{ab} = a^{-2}(t)\alpha(t) u_a u_b + a^{-2}(t)\beta(t) g_{ab}.
\end{equation}
where $\alpha(t)$ and $\beta(t)$ are smooth functions.
Similarly,
\begin{equation}
\label{eq_flrw_curv_scalar}
\begin{array}{lll}
R &=& g^{st}R_{st} \\
&=& \kappa\(-\rho c^2 - p + 2 \rho c^2 - 2 p\) + 4\Lambda \\
&=& \kappa\(\rho c^2 - 3p\) + 4\Lambda \\
\end{array}
\end{equation}
and there is a smooth function $\gamma(t)$ so that
\begin{equation}
\label{eq_flrw_curv_scalar_a}
R = a^{-2}(t)\gamma(t).
\end{equation}

We need to check that $a^{-2}(t)$ is compensated in $S_{abcd}$ and $E_{abcd}$, so that $a(t)$ appears to a non-negative power.

Since the {\FLRW} metric \eqref{eq_flrw_metric} is diagonal in the standard coordinates, each term in $(g\circ g)_{abcd}$ is of the form $g_{aa}g_{bb}$, with $a\neq b$. This means that at least $a\neq t$ or $b\neq t$ holds, and from \eqref{eq_flrw_metric} we conclude that $g_{aa}g_{bb}$ contains $a(t)$ at least to the power $2$. Therefore, the scalar part of the Riemann curvature, $S_{abcd}$, is smooth.

For the same reason, the Kulkarni-Nomizu product between the metric tensor and the term $a^{-2}(t)\beta(t) g_{ab}$ from the expression of the Ricci curvature \eqref{eq_flrw_curv_ricci_a} is smooth.

The only term from \eqref{eq_flrw_curv_ricci_a} we have to check that is smoothened by the Kulkarni-Nomizu product with $g$ is $a^{-2}(t)u_a u_b$. Since $u_a=g_{as}u^s$ and $g$ is diagonal, it follows that $u_a=g_{aa}u^a$ (without summation). If $b\neq t$ ($a\neq t$  is similar), then $u_b=g_{bb}u^b$ contains the needed $a^2(t)$. In the case when $a=b=t$, $a^{-2}(t) u_t u_t$ is not necessarily smooth, but in the Kulkarni-Nomizu product it will appear only in terms of the form $a^{-2}(t) u_t u_t g_{cc}$, with $c\neq t$. Hence, $\Ric\circ g$ is smooth. From this and from the smoothness of $S_{abcd}$, it follows that $E_{abcd}$ is also smooth.

One of the properties of the {\FLRW} metric is that it is conformally flat, that is, $C_{abcd} = 0$. From this it follows that $R_{abcd}=S_{abcd}+E_{abcd}$ is smooth too.
\end{proof}

\section{Properties of the proposed equation}
\label{s_properties}

\subsection{Conservation of energy}

The conservation of energy is usually put in the form
\begin{equation}
	-a^3\dot\rho = 3 a^2 \dot a \rho + \dsfrac 3 {c^2} a^2\dot a p,
\end{equation}
which remains valid even when the volume $a^3\to 0$.

\subsection{The metric is parallel}

It is known that if the metric tensor is regular, its covariant derivative vanishes, $g_{ab;c}=0$. For our solution, this is true so long as $a(t)\neq 0$. But if $a=0$, the metric is degenerate, and we have to check that $g_{ab;c}=0$.

The metric being diagonal, its Christoffel symbols of the first kind,
\begin{equation}
	\Gamma_{abc}=\dsfrac 1 2 \(g_{bc,a} + g_{ca,b} - g_{ab,c}\)
\end{equation}
which don't vanish are either of the form
\begin{equation}
	\Gamma_{aaa}=\dsfrac 1 2 g_{aa,a}
\end{equation}
or, for $a\neq b$,
\begin{equation}
	\Gamma_{aab} = -\dsfrac 1 2 g_{aa,b}
\end{equation}
or
\begin{equation}
	\Gamma_{aba} = \Gamma_{baa} = \dsfrac 1 2 g_{aa,b}
\end{equation}
Consequently, the Christoffel symbols of the second kind,
\begin{equation}
	\Gamma^c_{ab}=g^{cs}\dsfrac 1 2 \(g_{bs,a} + g_{sa,b} - g_{ab,s}\)
\end{equation}
are of the form
\begin{equation}
	\Gamma^a_{aa}=\dsfrac 1 2 \dsfrac {g_{aa,a}}{g_{aa}} (!)
\end{equation}
or, for $a\neq b$,
\begin{equation}
	\Gamma^b_{aa} = -\dsfrac 1 2 \dsfrac {g_{aa,b}}{g_{bb}} (!)
\end{equation}
or
\begin{equation}
	\Gamma^a_{ab} = \Gamma^a_{ba} = \dsfrac 1 2 \dsfrac {g_{aa,b}}{g_{aa}} (!)
\end{equation}
 where $(!)$ means ``no summation over the repeated indices''.

It follows that the covariant derivative introduces in the worst case a division by $a^2(t)$.

The covariant derivative of the metric tensor is
\begin{equation}
	g_{ab;c} = g_{ab,c} - \Gamma^s_{bc}g_{as} - \Gamma^s_{ac}g_{sb}.
\end{equation}
Obviously $g_{ab,c}$ is smooth, because $g_{ab}$ is smooth. From the other terms, the only ones involving non-vanishing Christoffel symbols are of the form $\Gamma^a_{aa}g_{aa}$, $\Gamma^a_{bb}g_{aa}$, and $\Gamma^a_{ab}g_{aa}$, without summation. Whenever $\Gamma^a_{bc}$ involves $a(t)$ to a negative power, which can only be $1$ or $2$, this is compensated by $g_{aa}$, which contains $a^2(t)$. It follows that the covariant derivative of the metric tensor is smooth, and by continuity is zero even when the metric becomes degenerate (at $a(t)=0$):
\begin{equation}
\label{eq_parallel_metric}
	g_{ab;c} = 0.
\end{equation}

\subsection{The Bianchi identity}

We will show that the Riemann curvature tensor satisfies the Bianchi identity
\begin{equation}
	R_{(abc)d;e} = 0.
\end{equation}
Given that it holds at all the points for which $a(t)\neq 0$, where the metric is regular, it also holds by continuity at $a(t)= 0$. But we need to check that the covariant derivatives $(\Ric\circ g)_{abcd;e}$ are smooth, because if the Bianchi identity would be between infinte values, there would be no continuity.

Since the Weyl part of the Riemann curvature $C_{abcd}=0$ in the {\FLRW} spacetime, and from the equations \eqref{eq_flrw_curv_ricci_a} and \eqref{eq_flrw_curv_scalar_a}, it follows that $R_{abcd}$ has the following form:
\begin{equation}
	R_{abcd} = a^{-2}(t)\mu(t) \((u\otimes u)\circ g\)_{abcd} + a^{-2}(t)\nu(t) (g\circ g)_{abcd},
\end{equation}
where the functions $\mu(t)$ and $\nu(t)$ are smooth.

Let's denote by $h_{ij}$, $1\leq i,j\leq 3$, the metric on $\Sigma$. Then $g_{ij}=a^2(t)h_{ij}$. Given that our frame is comoving with the fluid, $u_t=1$, and $u_i=0$ for all $i$. The only terms $a^{-2}(t)\mu(t) \((u\otimes u)\circ g\)_{abcd}$ which don't vanish by containing $u_i=0$ are of the form $a^{-2}(t)\mu(t)u_tu_tg_{ii}$. The covariant derivatives with respect to $t$ cancel one another in the Bianchi identity under the permutation, because the index $t$ is repeated. So we check now those terms of the form $\nabla_j\(a^{-2}(t)\mu(t)u_tu_tg_{ii}\)$, where $i\neq j$. But $\nabla_j\(a^{-2}(t)\mu(t)u_tu_tg_{ii}\)=a^{-2}(t)\mu(t)u_tu_tg_{ii;j}=0$, because only $g_{ii}$ depends on the spacelike direction $x^j$, and because the metric tensor is parallel \eqref{eq_parallel_metric}.

The terms $a^{-2}(t)\nu(t) (g\circ g)_{abcd}$ can only be of the form $a^{-2}(t)\nu(t) g_{aa}g_{bb}$, $a\neq b$. Since $\nabla_i\(a^{-2}(t)\nu(t) g_{aa}g_{bb}\) = a^{-2}(t)\nu(t) \(g_{aa;i}g_{bb} + g_{aa}g_{bb;i}\) = 0$, it follows that the covariant derivatives with respect to $i$ vanish. We check now thosw with respect to $t$. If either the index $a$ or $b$ is equal to $t$, then the cyclic permutation involved in the Bianchi identity vanishes. Then the only remaining possibility is $\nabla_t\(a^{-2}(t)\nu(t) g_{ii}g_{jj}\)$. But $\nabla_t\(a^{-2}(t)\nu(t) g_{ii}g_{jj}\) = -2\dot a(t)a^{-3}(t)\nu(t) g_{ii}g_{jj} + a^{-2}(t)\dot\nu(t) g_{ii}g_{jj}= -2\dot a(t)a(t)\nu(t) h_{ii}h_{jj} + a^{2}(t)\dot\nu(t) h_{ii}h_{jj}$, which is smooth.

Hence, the Bianchi identity makes sense even at the singularity $a(t)=0$.

\subsection{Action principle}

Shortly after Einstein proposed his field equation, Hilbert and Einstein provided a Lagrangian formulation. The Lagrangian density which leads to Einstein's equation with matter given by $\mc L\sqrt{-g}$ and cosmological constant $\Lambda$ is
\begin{equation}
\label{eq_lagrangian}
	\dsfrac{1}{2\kappa}\(R\sqrt{-g} - 2\Lambda\sqrt{-g}\) + \mc L\sqrt{-g}.
\end{equation}
In our case, the scalar curvature is singular at $a(t)\to 0$. But this doesn't affect the Lagrangian density, since the density $R\sqrt{-g}$ is smooth \cite{Sto11h}.
Given that our expanded Einstein equation \eqref{eq_einstein_expanded} is equivalent to Einstein's its solutions are extremals of the action given by \eqref{eq_lagrangian}.

\section{Conclusions}
\label{s_conclusions}

The new form of Einstein's equation extends uniquely  beyond the Big Bang singularity, as it is represented schematically in Figure \ref{flrw-exp}.

\image{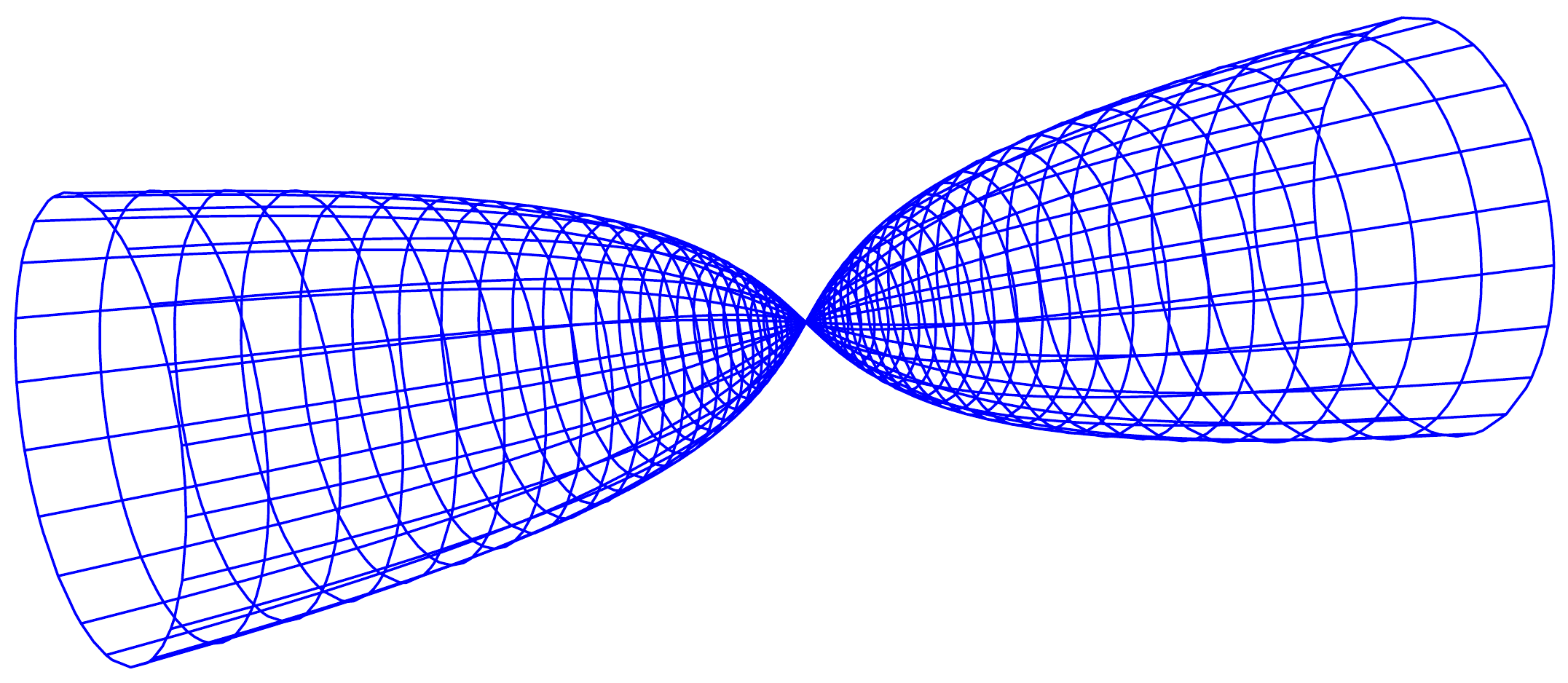}{0.8}{Schematic representation of a generic {\FLRW} spacetime. The solutions of the new equation can be continued naturally before the Big Bang.}

An alternative solution was proposed in \cite{Sto11h}, where the Einstein equation was replaced with a densitized version
\begin{equation}
\label{eq_einstein_idx:densitized}
	G_{ab}\sqrt{-g} + \Lambda g_{ab}\sqrt{-g} = \kappa T_{ab}\sqrt{-g}.
\end{equation}
This version is expressed in terms of tensor densities of weight $1$, which appear naturally from the Lagrangean.

The {\FLRW} spacetime is an ideal one, based on the assumptions of the \textit{cosmological principle} (that it is homogeneous and isotropic). But the extension proposed here opens new possibilities to explore.

Singularities which are of the type studied here, having the Riemann curvature tensor $R_{abcd}$ smooth, and admitting smooth Ricci decomposition, are in fact more general. In \cite{Sto11e} the Schwarzschild singularity is put in a form in which has these properities, by an appropriate coordinate change. This, and similar results on the Reissner-Nordstr\"om singularity \cite{Sto11f} suggests that we should reconsider the information loss \cite{Sto12e}. More general cosmological models, which are neither homogeneous nor isotropic, are studied in \cite{Sto12c}, and shown to admit a smooth Ricci decomposition, and satisfy the Weyl curvature hypothesis \cite{Pen79}. Implications suggesting to reconsider the the problem of quantization are presented in \cite{Sto12d,Sto12f}.

\textbf{Acknowledgments}

I thank an anonymous referee for the valuable comments and suggestions to improve the clarity and the quality of this paper.

\bibliographystyle{unsrt}

\end{document}